\documentclass[11pt]{article}
\usepackage[margin=1in]{geometry}
\usepackage{amsmath,amssymb,amsthm,mathtools}
\usepackage{microtype}
\usepackage[colorlinks=true,linkcolor=blue,citecolor=blue,urlcolor=blue]{hyperref}
\usepackage[nameinlink,capitalize]{cleveref}
\hypersetup{pdftitle={Sharp High-Entropy Bounds for Sums of Independent Discrete Random Variables},
pdfauthor={Haoran Wang},pdfkeywords={Shannon entropy, entropy power inequality, rearrangement, additive combinatorics}}
\numberwithin{equation}{section}
\newtheorem{theorem}{Theorem}[section]
\newtheorem{corollary}[theorem]{Corollary}
\newtheorem{lemma}[theorem]{Lemma}
\newcommand{\F}{\mathbb F}
\newcommand{\Z}{\mathbb Z}
\newcommand{\R}{\mathbb R}
\newcommand{\Prob}{\mathbb P}
\newcommand{\E}{\mathbb E}

\title{Sharp High-Entropy Bounds for Sums\\of Independent Discrete Random Variables}
\author{Haoran Wang\\\small Independent Researcher\\
\small\texttt{whr.hrwang@gmail.com}}
\date{}

\begin{document}
\maketitle

\begin{abstract}
Sharp high-entropy lower bounds for the entropy of a sum were known for identically distributed summands in torsion-free abelian groups and in prime cyclic groups. For arbitrary independent summands, Gavalakis, Goh and Kontoyiannis obtained an additive constant of $1/8$ and conjectured that the sharp constant is $1/2$. We prove that independent discrete random variables $X,Y$ with finite Shannon entropies in bits satisfy $H(X+Y)\ge (H(X)+H(Y))/2+1/2-o(1)$ in every torsion-free abelian group as $\max\{H(X),H(Y)\}\to\infty$. The same conclusion holds in the prime cyclic group $\F_p$ when both $\max\{H(X),H(Y)\}$ and $\log_2p-\max\{H(X),H(Y)\}$ tend to infinity. We give explicit error bounds in both settings. The proof extracts a component with paired point probabilities while controlling the entropy of the remainder independently of its support. Discrete rearrangement and uniform perturbation then transfer the continuous entropy power inequality to this component. In prime cyclic groups, an additional estimate controls the entropy lost under modular reduction. Binomial distributions show that the constant $1/2$ is optimal.
\end{abstract}

\section{Introduction}

Entropy inequalities for sums describe how uncertainty grows under addition of independent random variables. For discrete $X,Y$ in an abelian group, Shannon entropy satisfies $H(X+Y)\ge\max\{H(X),H(Y)\}$. We seek a stronger lower bound in terms of the two input entropies. All logarithms are to base $2$, so entropy is measured in bits.

The continuous entropy power inequality of Shannon and Stam provides the natural benchmark \cite{Shannon1948,Stam1959}. For independent real random variables with densities, it gives $2^{2h(X+Y)}\ge 2^{2h(X)}+2^{2h(Y)}$, where $h$ denotes differential entropy and the entropies are assumed finite. The arithmetic--geometric mean inequality therefore yields
\[
h(X+Y)\ge\frac{h(X)+h(Y)}2+\frac12.
\]
Independent Gaussians with equal variances attain equality. A positive additive constant cannot hold for discrete variables at every entropy scale, since two deterministic summands have a deterministic sum. The question is whether the same half-bit bound emerges as the input entropy grows.

We answer this question for arbitrary independent summands in torsion-free abelian groups and prime cyclic groups. An abelian group is torsion-free if $mg=0$, for a positive integer $m$, implies $g=0$; examples include $\Z^d$ and the additive group of $\R$. For prime $p$, the additive group of the field $\F_p$ is cyclic of order $p$. In this finite setting, the input entropy must also stay far enough below the uniform entropy $\log_2p$.

\subsection{Prior work}

Entropy versions of sumset inequalities connect this problem to additive combinatorics. Tao developed sumset and inverse-sumset theory for Shannon entropy and proved the high-entropy half-bit bound for independent identically distributed variables in torsion-free abelian groups \cite{Tao2010}. Haghighatshoar, Abbe and Telatar obtained universal entropy-growth bounds for integer-valued variables, including summands with different distributions \cite{HAT2014}. Their proof separates distributions with a large point probability from those with small point probabilities. This distinction also enters our argument, but their general bound does not reach the sharp constant $1/2$.

Two further developments supply the tools for our proof. Madiman, Wang and Woo established entropy inequalities through rearrangement in prime cyclic groups \cite{MWW2021}. Their comparison preserves the individual entropies and lowers the sum entropy when one summand has a suitable paired form. They also decompose general distributions into regular components. A different approach adds continuous noise and applies the continuous entropy power inequality: Nekouei, Skoglund and Johansson used this method for discrete variables \cite{Nekouei2019}, and Gavalakis obtained quantitative entropy monotonicity for log-concave sums using uniform perturbations \cite{Gavalakis2024}. The mixture-entropy inequality of Melbourne, Talukdar, Bhaban, Madiman and Salapaka relates the effect of such perturbations to total variation distance \cite{Melbourne2022}. We combine these tools with a regularization whose discarded entropy is controlled by its mass alone.

The ambient group also matters. Jog and Anantharam studied entropy-power inequalities for groups of order $2^n$ \cite{JogAnantharam2014}, while Gavalakis, Kontoyiannis, Sriramu and Wagner recently proved equality and stability results on compact abelian groups \cite{GKSW2026}. Proper subgroups obstruct a universal positive entropy increment. In a prime cyclic group this obstruction is absent, but concentration near the uniform distribution still prevents a half-bit gain.

Gavalakis, Goh and Kontoyiannis proved the prime-field counterpart of Tao's theorem and then treated arbitrary independent summands over $\R$ and $\F_p$ \cite{GGK2026}. Under high-entropy assumptions, together with a deficit from uniform entropy in the prime-field case, they obtained
\[
H(X+Y)\ge\frac{H(X)+H(Y)}2+\frac18-\varepsilon.
\]
Their comparison with identically distributed summands loses a factor of four. They conjectured that the sharp constant is $1/2$.

\subsection{Our results}

For $0\le t\le1$, let $h_2(t)=-t\log_2t-(1-t)\log_2(1-t)$ denote the binary entropy function, with $0\log_2 0=0$.

\begin{theorem}\label{thm:main}
Let $X,Y$ be independent discrete random variables with finite entropies, and put $M=\max\{H(X),H(Y)\}>1$.
If $X,Y$ take values in a torsion-free abelian group, then
\begin{equation}\label{eq:main-tf}
\begin{aligned}
H(X+Y)\ge{}&\frac{H(X)+H(Y)}2+\frac12
-\frac{3\log_2e}{4(M-1)}\\
&-\frac34\log_2\left(1+2^{-(M-1)/2}\right).
\end{aligned}
\end{equation}
If $X,Y$ take values in $\F_p$, where $p$ is an odd prime, and $K=\log_2p-M\ge9$, then
\begin{equation}\label{eq:main-prime}
\begin{aligned}
H(X+Y)\ge{}&\frac{H(X)+H(Y)}2+\frac12
-\frac{3\log_2e}{4(M-1)}
-\frac34\log_2\left(1+2^{-(M-1)/2}\right)\\
&-h_2\left(\frac4{K-1}\right)-\frac4{K-1}.
\end{aligned}
\end{equation}
\end{theorem}

Thus the error vanishes as $M\to\infty$ and, in the prime-field case, $K\to\infty$. The conjecture of Gavalakis, Goh and Kontoyiannis \cite{GGK2026} imposes conditions on just one summand. It follows in the form below, since an entropy difference of at least one bit already gives the required bound.

\begin{corollary}\label{cor:ggk}
For every $\varepsilon>0$ there are finite constants $h_\varepsilon,K_\varepsilon$ with the following property. Let $X,Y$ be independent discrete random variables with finite entropies and $H(X)\ge h_\varepsilon$. If they take values in a torsion-free abelian group, or in $\F_p$ with $H(X)\le\log_2p-K_\varepsilon$, then
\[
H(X+Y)\ge\frac{H(X)+H(Y)}2+\frac12-\varepsilon.
\]
\end{corollary}

The constant is optimal already for binomial distributions. For independent $X_n,Y_n$ with distribution $\mathrm{Binomial}(n,1/2)$, the entropy asymptotic in \cite{GK2024} gives
\[
H(X_n)=\frac12\log_2\left(\frac{\pi e n}{2}\right)+o(1),
\qquad H(X_n+Y_n)-H(X_n)=\frac12+o(1).
\]
The same example works in $\F_{p_n}$ for primes $p_n>2n$: reduction modulo $p_n$ preserves all three entropies, and $\log_2p_n-H(X_n)\to\infty$. For general finite fields, proper additive subspaces obstruct a positive universal increment even for full-support distributions; we give the short construction in \cref{sec:discussion}.

\subsection{Technique overview}

The continuous entropy power inequality supplies the half-bit gain, but applying it to arbitrary discrete distributions requires two reductions. High entropy alone gives no useful control on the largest point probability or the shape of a distribution. However, if the sum entropy is less than half a bit above the average input entropy, conditioning on a most likely value forces the largest point probability of the higher-entropy summand to be $O(M^{-1})$. This is precisely the regime in which a discrete-to-continuous comparison has a small error.

The main difficulty is to use rearrangement without losing the half-bit gain. We extract from the other summand a component whose ordered probabilities occur in equal pairs after the first entry. If the remainder has mass $r$, its entropy contribution is at most $h_2(2r)/2$, independently of the support size. The remainder is still useful: on that branch the sum retains the entire entropy of the higher-entropy input. This contribution offsets the cost of the decomposition, leaving a net loss of order $2^{-M/2}$. The rearrangement inequality of Madiman, Wang and Woo then reduces the retained branch to centered integer distributions without further loss. Adding independent uniform noise on $[0,1]$ preserves each input entropy as differential entropy. The mixture bound of Melbourne et al. controls the cost of removing the noise by the largest point probability, giving the remaining $O(M^{-1})$ error.

For prime cyclic groups, the integer sum must also be reduced modulo $p$. Under the same small-increment assumption, the entropy deficit $K=\log_2p-M$ bounds the mass outside the most probable half of the values by $O(K^{-1})$. Centered rearrangement turns this rank bound into a bound on the probability of a modular carry. Its entropy contributes an additional $O((\log K)/K)$ loss. Section~\ref{sec:torsion-free} proves the regularization and entropy transfer, including the passage to torsion-free groups. Section~\ref{sec:prime} controls this final modular loss.Concrete choices of the thresholds in Corollary~\ref{cor:ggk} are recorded in Appendix~\ref{app:thresholds}.

\section{Preliminaries}\label{sec:preliminaries}

All discrete random variables below have finite Shannon entropy. For a probability mass function $q$, write $H(q)=-\sum_xq(x)\log_2q(x)$ and $H(X)=H(q)$ when $X$ has mass function $q$. For a density $f$ on $\R$, write $h(f)=-\int_{\R}f(x)\log_2f(x)\,dx$, and use $h(X)$ when $X$ has density $f$. We write $\mathcal L(X)$ for the distribution of $X$, and $\Prob$ and $\E$ for probability and expectation. An atom is a value of positive probability; its mass is that point probability. The support consists of all atoms. We use the usual joint and conditional entropies $H(X,Y)$ and $H(X\mid Y)$, and write $H(X\mid E)$ for conditioning on an event $E$ of positive probability. For independent variables in an abelian group,
\begin{equation}\label{eq:addmon}
H(X+Y)\ge\max\{H(X),H(Y)\},
\end{equation}
because $H(X+Y\mid Y)=H(X)$ and conditioning cannot increase entropy.

For mass functions $q,s$ on a common countable set, their total variation distance is $d_{\mathrm{TV}}(q,s)=\frac12\sum_x|q(x)-s(x)|$. The two-component mixture bound of Melbourne et al.\ \cite{Melbourne2022} is
\begin{equation}\label{eq:mixture-tv}
H(tq+(1-t)s)-tH(q)-(1-t)H(s)
\le h_2(t)d_{\mathrm{TV}}(q,s),\qquad 0\le t\le1.
\end{equation}
For completeness, put $d=d_{\mathrm{TV}}(q,s)$. When $0<d<1$, the two laws share a common component of mass $1-d$, with normalized mass function $\min\{q,s\}/(1-d)$. Choose between $q$ and $s$ with probabilities $t,1-t$, and reveal whether the sample came from this common component or from the residual component. The former gives no information about the choice; the latter gives at most $h_2(t)$ bits and occurs with probability $d$. The left side of \eqref{eq:mixture-tv} is exactly the information about that choice contained in the sample. The endpoints $d=0,1$ follow directly. We also use
\begin{equation}\label{eq:bern-var}
h_2(u)-cu\le\log_2(1+2^{-c}),\qquad c\ge0,\quad 0\le u\le1,
\end{equation}
which follows by maximizing over $u$.

For a finitely supported distribution, arrange its probabilities $a_1\ge a_2\ge\cdots$ in the order $0,1,-1,2,-2,\ldots$, padding with zeros. Write $A^\circ$ for a variable with this centered rearrangement of the law of $A$. The rearrangement preserves entropy and is unimodal: its probabilities are nondecreasing up to a mode and nonincreasing thereafter. On $\F_p$ we use the same order on the centered representatives $-(p-1)/2,\ldots,(p-1)/2$. Following Madiman, Wang and Woo, a distribution is $\triangle$-regular if $a_2=a_3$, $a_4=a_5$, and so on. Its centered rearrangement is symmetric about zero. All rearranged variables appearing together in a sum are taken to be independent.

The rearrangement inequality in \cite{MWW2021} implies that, if $B$ is $\triangle$-regular, then
\begin{equation}\label{eq:mww}
H(A+B)\ge H(A^\circ+B^\circ)
\end{equation}
for variables on $\F_p$, with both sums taken modulo $p$. It also holds for finitely supported integer variables: choose a prime large enough that reduction is injective on the supports and sumsets before and after rearrangement, and apply the prime-cyclic comparison. Here the sumset of two sets consists of all sums with one element from each set.

\section{Regularization and the torsion-free bound}\label{sec:torsion-free}

\subsection{A paired component with controlled residual entropy}

The regular components in \cite{MWW2021} suggest extracting a paired part before rearranging. The estimate needed here concerns the entropy carried by what is removed. A bound in terms of the support size would not suffice, since a small mass can be spread over arbitrarily many values. The following construction controls that entropy through the ranks of the removed probabilities.

\begin{lemma}\label{lem:pair}
Let $Y$ be finitely supported. There exist $0\le r\le1/2$ and variables $Y_0,Y_1$ such that
\[
\mathcal L(Y)=(1-r)\mathcal L(Y_0)+r\mathcal L(Y_1),
\qquad rH(Y_1)\le\frac12h_2(2r),
\]
where $Y_0$ is $\triangle$-regular. In particular,
\begin{equation}\label{eq:regularization-entropy}
(1-r)H(Y_0)\ge H(Y)-h_2(r)-\frac12h_2(2r).
\end{equation}
Moreover, for every positive integer $m$, the mass of $Y_0$ outside its $m$ largest probabilities, multiplied by $1-r$, is at most the mass of $Y$ outside its $m$ largest probabilities.
\end{lemma}

\begin{proof}
Write the probabilities of $Y$ as $q_1\ge q_2\ge\cdots$, padded with zeros, and retain the sequence $q_1,q_3,q_3,q_5,q_5,\ldots$ at the original values. The removed masses are $d_j=q_{2j}-q_{2j+1}$ for $j\ge1$, with total $r=\sum_{j\ge1}d_j\le\sum_{j\ge1}q_{2j}\le1/2$. Dividing the retained masses by $1-r$ defines $Y_0$. If $r>0$, assign probability $d_j/r$ to the original value at rank $2j$ to define $Y_1$. The retained sequence is still nonincreasing and is bounded term by term by the original sequence, which proves both its paired form and the assertion about ranks. When $r=0$, take $Y_0$ to have the law of $Y$ and $Y_1$ deterministic.

Suppose $r>0$, and let $J$ have probabilities $\Prob(J=j)=d_j/r$ on the positive integers. The essential estimate is
\[
\begin{aligned}
r\E J
&=\sum_{j\ge1}j(q_{2j}-q_{2j+1})
\le\sum_{j\ge1}j(q_{2j}-q_{2j+2})\\
&=\sum_{j\ge1}q_{2j}\le\frac12.
\end{aligned}
\]
The geometric distribution maximizes entropy on the positive integers under a mean constraint. Thus $\E J\le1/(2r)$ gives $H(J)\le h_2(2r)/(2r)$. This also holds at $r=1/2$, when the mean bound forces $J=1$. Since $Y_1$ is determined by $J$, we obtain $rH(Y_1)\le h_2(2r)/2$. Finally, revealing the mixture component gives $H(Y)\le h_2(r)+(1-r)H(Y_0)+rH(Y_1)$, which proves \eqref{eq:regularization-entropy}.
\end{proof}

\subsection{Entropy transfer}

Uniform perturbation converts the entropy of an integer-valued input exactly into differential entropy. For a sum, the perturbation mixes its mass function with a unit shift. The mixture inequality \eqref{eq:mixture-tv} therefore gives a quantitative form of the perturbation argument in \cite{Nekouei2019,Gavalakis2024}.

\begin{lemma}\label{lem:dither}
Let $A,B$ be independent finitely supported integer-valued variables, with $A$ unimodal, and put $a=\max_x\Prob(A=x)$. Then
\begin{equation}\label{eq:dither}
H(A+B)\ge\frac12\log_2\left(2^{2H(A)}+2^{2H(B)}\right)
-\frac{\log_2e}{2}a.
\end{equation}
\end{lemma}

\begin{proof}
Let $U,V$ be independent uniforms on $[0,1]$, independent also of $A,B$. The unit intervals carrying the densities of their integer translates have disjoint interiors, so $h(A+U)=H(A)$ and $h(B+V)=H(B)$. Put $S=A+B$, with mass function $s=(s_k)_{k\in\Z}$, where $s_k=\Prob(S=k)$, and let $\tau s$ be its unit shift, $(\tau s)_k=s_{k-1}$. The density of $S+U+V$ at $k+t$, for $0<t<1$, is $ts_k+(1-t)s_{k-1}$. Consequently, \eqref{eq:mixture-tv} gives
\[
\begin{aligned}
h(S+U+V)-H(S)
&=\int_0^1\left[H(ts+(1-t)\tau s)-H(s)\right]dt\\
&\le d_{\mathrm{TV}}(s,\tau s)\int_0^1h_2(t)\,dt
=\frac{\log_2e}{2}d_{\mathrm{TV}}(s,\tau s).
\end{aligned}
\]
Adding an independent variable contracts total variation, so $d_{\mathrm{TV}}(s,\tau s)\le d_{\mathrm{TV}}(\mathcal L(A),\mathcal L(A+1))=a$. The last equality follows by telescoping on either side of a mode; the same overlap identity is used in \cite{Gavalakis2024}. Applying the continuous entropy power inequality to $A+U$ and $B+V$ proves \eqref{eq:dither}.
\end{proof}

The next estimate combines regularization and rearrangement with this comparison. We state it for integers and prime cyclic groups together, retaining the entropy loss under modular reduction for use in Section~\ref{sec:prime}.

\begin{lemma}\label{lem:transfer}
Let $X,Y$ be independent finitely supported variables in $\Z$ or $\F_p$, where $p$ is an odd prime. Suppose $H(X)=M\ge H(Y)$, $M>1$, and
\begin{equation}\label{eq:small-increment}
H(X+Y)-\frac{H(X)+H(Y)}2<\frac12.
\end{equation}
Apply \cref{lem:pair} to $Y$, and let $A=X^\circ$, $B=Y_0^\circ$, regarded as integer-valued variables. Put $w=0$ in the integer case. In the prime-field case, let $T_p$ be the centered reduction of $A+B$ modulo $p$ and put $w=H(A+B)-H(T_p)$. Then
\begin{equation}\label{eq:transfer}
\begin{aligned}
H(X+Y)-\frac{H(X)+H(Y)}2
\ge{}&\frac12-\frac{3\log_2e}{4(M-1)}
-\frac34\log_2\left(1+2^{-(M-1)/2}\right)\\
&-(1-r)w.
\end{aligned}
\end{equation}
\end{lemma}

\begin{proof}
By \eqref{eq:addmon} and \eqref{eq:small-increment}, we have $H(Y)>M-1$ and $H(X+Y)-M<1/2$. Choose a most likely value $x_0$ of $X$, with probability $a$. Conditioning on whether $X=x_0$, and using \eqref{eq:addmon} on the complementary branch, yields
\[
\begin{aligned}
H(X+Y)&\ge aH(Y)+(1-a)H(X\mid X\ne x_0)\\
&=M+aH(Y)-h_2(a).
\end{aligned}
\]
It follows that $aH(Y)<1/2+h_2(a)\le3/2$, and hence
\begin{equation}\label{eq:atom}
a\le\frac{3}{2(M-1)}.
\end{equation}
This conditioning argument is the large-point-probability mechanism used in \cite{HAT2014}.

Take the mixture components independently of $X$. Rearrangement \eqref{eq:mww} and \cref{lem:dither}, followed by the arithmetic--geometric mean inequality, give
\[
H(X+Y_0)\ge\frac{M+H(Y_0)}2+\frac12-\frac{\log_2e}{2}a-w.
\]
On the residual branch, $H(X+Y_1)\ge M$. Conditioning on the branch and substituting \eqref{eq:regularization-entropy} therefore gives
\[
\begin{aligned}
H(X+Y)-\frac{M+H(Y)}2
\ge{}&\frac12-\frac{\log_2e}{2}a+\frac{r(M-1)}2\\
&-\frac12h_2(r)-\frac14h_2(2r)-(1-r)w.
\end{aligned}
\]
Here we dropped the nonnegative term $ra\log_2e/2$. The positive term proportional to $rM$ compensates for the entropy cost of the decomposition.

Put $c=(M-1)/2$. Applying \eqref{eq:bern-var} at $r$ and $2r$ gives
\[
\begin{aligned}
\frac12h_2(r)+\frac14h_2(2r)-cr
&=\frac12\bigl[h_2(r)-cr\bigr]
+\frac14\bigl[h_2(2r)-2cr\bigr]\\
&\le\frac34\log_2(1+2^{-c}).
\end{aligned}
\]
Together with \eqref{eq:atom}, this proves \eqref{eq:transfer}.
\end{proof}

\begin{proof}[Proof of \cref{thm:main} in torsion-free groups]
For finitely supported integer variables, either the increment is at least $1/2$, or \cref{lem:transfer} applies with $w=0$ after relabeling. This proves \eqref{eq:main-tf} on $\Z$. Finite supports in a torsion-free abelian group generate a subgroup isomorphic to $\Z^d$ for some nonnegative integer $d$. Let $S$ be the union of the two supports and their sumset. An integer linear functional $\phi:\Z^d\to\Z$ can be chosen injective on $S$: its coefficient vector need only avoid the finitely many hyperplanes orthogonal to nonzero differences of elements of $S$. Since $\phi$ is additive and injective on these sets, it preserves $H(X)$, $H(Y)$ and $H(X+Y)$. The finite-support result follows.

For countable supports, choose increasing finite sets $A_n,B_n$ whose probabilities tend to one, and let $X_n,Y_n$ be independent variables with the conditional laws of $X$ on $A_n$ and $Y$ on $B_n$. Finite entropy implies $H(X_n)\to H(X)$ and $H(Y_n)\to H(Y)$. The event $E_n=\{X\in A_n,Y\in B_n\}$ preserves independence under conditioning, and conditioning on its indicator gives $H(X+Y)\ge\Prob(E_n)H(X_n+Y_n)$. Apply the finite-support bound and let $n\to\infty$; the right side of \eqref{eq:main-tf} is continuous for $M>1$.
\end{proof}

\section{Controlling modular reduction}\label{sec:prime}

The remaining loss in \cref{lem:transfer} comes from identifying integer sums that differ by $p$. Centered rearrangement makes this possible only when one input lies outside a central interval. We control that probability through ranks in the original distributions, using the entropy deficit from uniformity.

\begin{lemma}\label{lem:tail}
Let $p$ be an odd prime and let $X,Y$ be independent $\F_p$-valued variables satisfying \eqref{eq:small-increment}, with $H(X)=M\ge H(Y)$ and $K=\log_2p-M>1$. Put $L=\lfloor(p-1)/4\rfloor$ and $m=2L+1$. If $t_X,t_Y$ are the masses outside the $m$ largest probabilities of $X,Y$, respectively, then
\[
t_X\le\frac2{K-1},\qquad t_Y\le\frac2{K-1}.
\]
\end{lemma}

\begin{proof}
The small-increment assumption implies $|H(X)-H(Y)|<1$, $H(X+Y)-H(X)<1/2$ and $H(X+Y)-H(Y)<1$. Suppose $t_Y>0$, since otherwise its bound is immediate. If $q_1\ge\cdots\ge q_p$ are the ordered probabilities of $Y$, then $q_{m+1}\le1/(m+1)<2/p$. Thus every conditional probability on the tail is at most $2/(pt_Y)$, and
\[
H(Y\mid\mathrm{tail})\ge\log_2p-1+\log_2t_Y.
\]
Condition on this tail event. On its complement the sum entropy is at least $H(X)$, and on the tail it is at least $H(Y\mid\mathrm{tail})$. Therefore
\[
H(X+Y)-H(X)\ge t_Y\bigl(K-1+\log_2t_Y\bigr).
\]
The left side is less than $1$, while $-t\log_2t<1$ for $0\le t\le1$. Hence $t_Y(K-1)<2$. Interchanging $X,Y$ gives the other bound, since $\log_2p-H(Y)\ge K$ and the corresponding entropy increment is also less than $1$.
\end{proof}

\begin{proof}[Proof of \cref{thm:main} in prime cyclic groups]
Assume the increment is less than $1/2$ and relabel so that $H(X)=M\ge H(Y)$. Use $r,A,B,T_p,w$ from \cref{lem:transfer}. Since $A,B$ are centered representatives, there is a unique carry $W\in\{-1,0,1\}$ such that $A+B=T_p+pW$. The pair $(T_p,W)$ determines the integer sum and is determined by it. Consequently, with $\theta=\Prob(W\ne0)$,
\[
w=H(W\mid T_p)\le H(W)\le h_2(\theta)+\theta.
\]
No carry occurs if $|A|,|B|\le L$, where $L$ is as in \cref{lem:tail}. The first of these bounds fails with probability $t_X$. The rank estimate in \cref{lem:pair} bounds the probability of failure for $B$ by $t_Y/(1-r)$. Thus
\[
(1-r)\theta\le(1-r)t_X+t_Y\le\frac4{K-1}.
\]
The weight $1-r$ must be kept here, because it is the weight of this branch in the original sum. Concavity of $h_2$ gives $(1-r)h_2(\theta)\le h_2((1-r)\theta)$. Since $K\ge9$, the last displayed upper bound is at most $1/2$, where $h_2$ is increasing. We obtain
\[
(1-r)w\le h_2\left(\frac4{K-1}\right)+\frac4{K-1}.
\]
Substitution in \eqref{eq:transfer} proves \eqref{eq:main-prime}. If the increment is at least $1/2$, the conclusion follows directly.
\end{proof}

\begin{proof}[Proof of \cref{cor:ggk}]
Put $M=\max\{H(X),H(Y)\}$. If $H(Y)\ge H(X)+1$, \eqref{eq:addmon} gives the conclusion. Otherwise $H(X)\le M<H(X)+1$, so in the prime-field case $\log_2p-M>\log_2p-H(X)-1$. Choosing $h_\varepsilon$ and $K_\varepsilon$ large enough makes the errors in \cref{thm:main} at most $\varepsilon$.
\end{proof}

\section{Discussion}\label{sec:discussion}

The restriction to prime fields reflects a subgroup obstruction. To see that requiring full support does not remove it, regard the additive group of $\F_{p^2}$ as a two-dimensional vector space over $\F_p$. Choose complementary one-dimensional subspaces $V,W$, so every field element has a unique representation $v+w$. Let $U$ be uniform on $V$ and, independently, let $Z$ take the value zero with probability $1-\eta$ and each nonzero value of $W$ with probability $\eta/(p-1)$, where $\eta=(\log_2p)^{-2}$. For primes $p\ge3$, the variable $X=U+Z$ has full support and satisfies
\[
H(X)=\log_2p+H(Z),\qquad
H(Z)=h_2(\eta)+\eta\log_2(p-1)=o(1).
\]
If $Y$ is an independent copy of $X$ and $Z'$ an independent copy of $Z$, then $H(X+Y)=\log_2p+H(Z+Z')$. Hence $0\le H(X+Y)-H(X)\le H(Z)\to0$, although both $H(X)$ and $2\log_2p-H(X)$ diverge. A general finite-group analogue must therefore measure concentration on proper subgroup cosets, in addition to entropy and its deficit from full-group uniformity. The equality and stability results in \cite{GKSW2026} provide a related starting point for such a formulation.

The proof also leaves a stability question: what additive structure is forced when the high-entropy increment approaches $1/2$? The binomial example is one possibility, but the reduction to integers preserves only finitely many additive relations and does not identify a canonical geometry for near-extremizers. A stability theorem would have to account for these additive embeddings as well as the near-equality cases of the continuous entropy power inequality. The optimal rates of the errors in $M$ and $K$ remain another concrete question.

\section*{Acknowledgments}
The author used GPT-5.6 Sol to assist with the initial proof for integer-valued variables and with drafting the manuscript. The author posed the problem and guided the development of the argument, including its extension to torsion-free abelian groups and prime cyclic groups. AI assistance was also used to check calculations, examine possible gaps, and revise the exposition. The author reviewed the arguments and the final manuscript and takes responsibility for the mathematical claims, references, and presentation.

\appendix

\section{Explicit thresholds}\label{app:thresholds}

We record one convenient quantitative choice in Corollary~\ref{cor:ggk}. Assume \(0<\varepsilon<1/2\); for larger \(\varepsilon\), the conclusion follows directly from \(H(X+Y)\ge \max\{H(X),H(Y)\}\).

For the torsion-free estimate, put \(t=M-1\). If \(t\ge4\), then

$$
\log_2(1+2^{-t/2})
\le (\log_2 e)2^{-t/2}
\le \frac{\log_2 e}{t}.
$$

Hence the two error terms in \eqref{eq:main-tf} are together at most

$$
\frac{3\log_2 e}{2t}.
$$

In particular, \(t\ge 5/\varepsilon\) is sufficient to make their sum at most \(\varepsilon/2\).

In the prime cyclic case, write

$$
\delta=\frac4{K-1}.
$$

If

$$
K\ge 1+\frac{16}{\varepsilon}\log_2\frac8\varepsilon,
$$

then

$$
\delta\le
\frac{\varepsilon}{4\log_2(8/\varepsilon)}.
$$

Using

$$
h_2(\delta)+\delta
\le
\delta\log_2\frac{2e}{\delta},
$$

the additional prime-field error is at most \(\varepsilon/2\). Indeed, with \(b=\log_2(8/\varepsilon)>4\),

$$
\delta\log_2\frac{2e}{\delta}
\le
\frac{\varepsilon}{4b}
\log_2\frac{8eb}{\varepsilon}
\le
\frac{\varepsilon}{2}.
$$

Finally, in the nontrivial regime of Corollary~\ref{cor:ggk}, the larger input entropy differs from \(H(X)\) by less than one bit. Thus one may take, for example,

$$
h_\varepsilon
=
1+\frac5\varepsilon,
\qquad
K_\varepsilon
=
2+\frac{16}{\varepsilon}\log_2\frac8\varepsilon.
$$

\end{document}